\documentclass[12pt]{amsart}
\usepackage{a4wide}
\usepackage[dvips]{epsfig}
\usepackage{amscd}
\usepackage{amssymb}
\usepackage{amsthm}
\usepackage{amsmath}
\usepackage{times}
\usepackage{latexsym}
\usepackage{subcaption}

\newtheorem{theorem}{Theorem}
\theoremstyle{plain}

\newtheorem{proposition}[theorem]{Proposition}

\numberwithin{equation}{section}
\numberwithin{theorem}{section}

\newcommand{\E}{\ensuremath{\mathbb{E}}}
\def\e{{\mathrm{e}}}
\allowdisplaybreaks

\begin{document}
\title[Option Price Bounds]{Pricing options on forwards in energy markets:\\ the role of mean reversion's speed}
\date{\today }

\author[Schmeck]{Maren Diane Schmeck}
\address[Maren Diane Schmeck]{\newline
Department of Mathematics\newline
University of Cologne\newline
Weyertal 86--90\newline
50931 Cologne, Germany}
\email[]{mschmeck\@@math.uni-koeln.de}
\maketitle
\begin{abstract}

Consider the problem of pricing options on forwards in energy markets, when spot prices follow a geometric multi-factor model in which several rates of mean reversion appear. 
In this paper we investigate the role played by slow mean reversion  when pricing and hedging options. In particular, we determine both upper and lower bounds for the error one makes neglecting low rates of mean reversion in the spot price dynamics. 

{\bf Keywords:}  electricity spot prices,  multi-scale mean reversion, delivery period, options on forwards,   hedging, pricing error, upper and lower bounds.
\end{abstract}

\section{Introduction}

A special feature of electricity markets is the  mean reverting behaviour of the spot prices. Here, we want to investigate the role of this mean reversion on the pricing and hedging problem of European options on forwards. As the underlying forwards typically deliver the commodity over a period, one may expect that the effect  averages out in the option price, depending on the speed of the mean reversion and on the length of the delivery period. In this paper, we want to quantify this effect, and determine upper and lower bounds for the pricing error if one or some mean reverting components of the electricity spot price are not taken into account.

Electricity spot prices are known to exhibit various characteristics that are rarely observed in other markets. One of these features is  the multi-scale mean reverting behaviour. Well known are the so-called spikes, large price shocks, which mean revert very fast towards the original price level  and have a half life of about two days. 
However, slower mean reverting components are observable in the electricity markets, too. Meyer-Brandis and Tankov \cite{MBT} 
find that at most European exchanges the autocorrelating structure of the spot prices  can be described very well with a weighted sum of exponentials.   This autocorrelation structure arises precisely for spot price models that include a sum of independent mean reverting processes, where the mean reversion takes place with different rates. As a consequence, they suggest a spot price model based on a sum of independent Ornstein-Uhlenbeck processes, and find for the EEX that two or three mean reverting factors are suitable. Now, the classical two-factor model after Gibson and Schwartz \cite{GS} and Schwartz and Smith \cite{SS} consists of only one mean reverting component, giving the short term variations of the spot price. Additionally, a non-stationary long term component plays the role of a  stochastic level of mean reversion and reflects long term expectations, as for example of political developments. Combining the two approaches, we assume a multi-factor spot price model consisting of one non-stationary process  and a sum of independent Ornstein-Uhlenbeck processes. 

When fitting multi-factor models to data, the question arises of how to filter from the observed time series the different, itself not observable  components. This is essential step for  pricing derivatives.  In the literature, there are several approaches that consider models with two factors. Many of them focus on filtering the spike component, see  Meyer-Brandis and Tankov \cite{MBT},  Benth and Schmeck \cite{BS2}, or filtering jumps (Borovkova and Permana \cite{BP},  A\"it-Sahalia \cite{A}).  Also, Barlow et. al \cite{BGM} apply the technique of  Kalman filtering to electricity markets. Nevertheless, if three or more factors are involved, the problem becomes even more challenging and it is therefore natural to investigate the impact of the factors on option pricing and hedging.
  
For this purpose, we consider options of European type, where the underlying is a forward. 
The standard way to price forwards is to define its price at time $t$ with delivery at $T>0$ as conditional expectation of the spot price 
\begin{align}\label{eq:forwardint}
f(t,T)= \mathbb{E} _\mathbb{Q}[ S(T)|\mathcal{F}_t]\;,\quad 0 \leq t \leq T,
\end{align} 
for some pricing measure $\mathbb{Q}$. Though, electricity forwards typically deliver the underlying over a period, as a month, three months or a year. One alternative approach is therefore to incorporate the delivery period by defining the forward as expected average spot 
 \begin{align}\label{eq:avfor}F(t,T_1,T_2)=\mathbb{E}_\mathbb{Q}\left[\frac{1}{T_2-T_1}\int_{T_1}^{T_2}S(t) dt\Big|\mathcal{F}_t\right]\;.\end{align} If we take, as we do in this paper, an exponential  model for the spot price, we do not obtain closed form expressions for the forward price \eqref{eq:avfor}, which is a significant drawback when analysing option prices. However, we aim to  use  an exponential model since it is convenient in option pricing. Indeed, a suitable exponential framework leads to Black's formula \cite{B}, used in practice e.g. by EEX \cite{EEX}.  In order to meet the previous aim while still having a good analytical tractability, we thus follow the approach to model the delivery period by its midpoint and define the forward price by \eqref{eq:forwardint}. If we take additionally into account that for energy options the delivery period starts at (or shortly after) the time $\tau$ of the exercise of the option, the length of the delivery period is given by $2(T-\tau)$. Now, we would like to consider delivery periods of different length. Fixing $t$ and $\tau$, the delivery period increases considering the limit $T\to \infty$. 

In this setting, Benth and Schmeck \cite{BS} show that  a fast mean reverting component does not significantly influence the price of an option. There, focusing on the spikes, they assume that the innovations of the fast mean reverting component are given by a pure jump L\'evy process. By introducing jumps to the model, the option price is no longer given by Black's formula. Nevertheless, the authors provide in \cite{BS} an upper bound for the error one makes when pricing an option and neglecting the spike component. This bound is exponential and depends on the speed of mean reversion and the length of the delivery period. Hence, for a big mean reversion parameter and a reasonable long delivery period, the error is very small such that the fast mean reverting jumps in the spot price dynamics are not relevant for option pricing.  As a consequence, Black's formula gives a good approximation to the option price.

Here, we are concerned with the role of slow mean reversion in the underlying spot price dynamics to option pricing. Although it is hard to justify that the innovations of the spikes are Gaussian, one can always describe the drivers of the slower mean reverting components by Brownian motions. Since a slowly mean reverting component does not average out so fast in the forward price, one may expect to make an error in the option price, when the delivery period is comparably short. We are therefore interested in quantifying the minimum pricing error one has neglecting slowly mean reverting components. To this end we determine a lower bound for the error. We are even able to find an explicit analytical function for the asymptotic behaviour of the pricing error by providing both lower and upper bounds for it. We show that the pricing error takes the form of a sum of exponentially decaying terms, where the speed is determined through the mean reversion rate, as well as  the parameters of the long term component. Furthermore, we are interested in how the error  in the option price transfers to the delta hedging component.\\

This paper is organized as follows. In Chapter \ref{ch:prep} we introduce the model for the spot price dynamics, we derive the forward price dynamics and the option price. We then determine the pricing error of the option when neglecting some mean reverting components in Chapter \ref{ch:pricingerror}. Finally, Delta hedging is considered in Chapter \ref{ch:hedging}.

\section{The spot price dynamics  and implied forward and option prices}\label{ch:prep}

Consider a filtered probability space $(\Omega, \mathcal{F},\{ \mathcal{F}\}_{t\geq0}, \mathbb{P})$ and an equivalent probability measure $\mathbb{Q}$.   Let the electricity spot price under  $\mathbb{Q}$ be given by a multi-factor model  in the spirit of the two factor model after Gibson and Schwartz \cite{GS}, Schwartz and Smith \cite{SS} and Lucia and Schwartz \cite{LS}, and a sum of mean reverting factors as in Benth et al. \cite{BKMB}:
\begin{align}\label{eq:spot}
S(t)=\Lambda(t) \exp\{X(t) +\sum_{i\in I} Y_i(t)\}\;,
\end{align}  
where $\Lambda(t)$ is a deterministic and bounded seasonality function. The non-stationary component is assumed to be a drifted Brownian motion
\begin{align}\label{eq:X}
dX(t)= \mu dt + \sigma dB(t)\;,
\end{align}
where $B$ is a standard Brownian motion and  $\mu,\,\sigma\geq0$ are constant. The stationary factors are supposed to be Ornstein-Uhlenbeck processes
\begin{align}\label{eq:Y}
dY_i(t)&= -\beta_i Y_i(t) dt + \sigma_i dB_i(t) 
\end{align}
for independent standard Brownian motions $B_i$, also independent  of $B$, and constants $\beta_i>0$ and $ \sigma_i >0$ for $i\in I=\{1,..,n\}$. In \eqref{eq:Y}, we choose the drivers of the stationary components to be Gaussian  in contrast to Benth and Schmeck \cite{BS}, who use pure jump L\'evy processes. There, the main target are the spikes. Those are big sudden upwards movements that mean revert quickly and it is therefore essential  to choose a driver that allows for big price movements, that is jumps.  In this paper, the emphasis is on slowly mean reverting components and the drivers we choose here fulfil this task. \\
The stationary factor $X$ can be interpreted as a long term factor, giving the general level of the prices and reflecting long term expectations towards political decisions, improving technologies or the storage level in case of storable commodities. The stationary factors $Y_i$, $i\in I$  reflect the every day price variations due to supply and demand, and thus play the role of short term components.  \\
To price forwards, we  relate the spot price to the forward price and use the standard definition 
\begin{align}\label{eq:forward}
f(t,T)= \mathbf{E} _\mathbb{Q}[ S(T)|\mathcal{F}_t]\;.
\end{align}  For convenience, we state the spot price \eqref{eq:spot} directly under $\mathbb{Q}$ and  omit in the sequel the $\mathbb{Q}$ in the associated expectations $\mathbb{E}_\mathbb{Q}$. Note that due to the non-storability of electricity, $\mathbb{Q}$ does not have to be an equivalent martingale measure, but an equivalent measure only. Therefore, there is more than one choice for  $\mathbb{Q}$. For a detailed discussion about the pricing measure and risk premium in electricity markets, see Benth et al. \cite{BSBK-book}, Chapter 1.5.3. The following proposition gives the forward price and its dynamics. 
\begin{proposition}\label{prop:forward}The forward price $f(t,T)$ at $t\geq0$ with delivery at $T\geq t$ is given by
\begin{align}\label{eq:fexpl}
f(t,T)= h(t,T) \exp\left\{ X(t)+\sum_{i\in I} e^{-\beta_i (T-t)} Y_i(t) \right\}\;
\end{align}
where 
\begin{align*}
h(t,T)=\Lambda(T) \exp\left\{ \left(\mu  + \frac12 \sigma^2\right) (T-t) + \frac12 \sum_{i\in I} \frac{\sigma_i^2}{2\beta_i}\left(1- e^{-2\beta_i(T-t)}\right)\right\}\;.
\end{align*}
The dynamics of the process $t\mapsto f(t, T)$ for $t\leq T$ is given by
\begin{align}\label{eq:fdyn}
df(t,T)= f(t,T)\left\{\sigma dB(t)+ \sum_{i\in I} e^{-\beta_i(T-t)} \sigma_i dB_i(t) \right\}\;.
\end{align}
\end{proposition}
\begin{proof}
We have that
\begin{align*}
X(T)=X(t)+\mu(T-t)+\sigma(B(T)-B(t))\,,
\end{align*}
and
\begin{align*}
Y_i(T)=\e^{-\beta_i(T-t)}Y_i(t)+\sigma_i\int_t^T\e^{-\beta_i(T-s)}\,dB_i(s)
\end{align*}
By the $\mathcal{F}_t$-adaptedness of $X(t)$ and  $Y_i(t)$,  the independent increment property of the Brownian motion and
the independence between $B$ and $B_i$ , $i\in I$, we then find
\begin{align*}
f(t,T)&=\Lambda(T)\E[\exp(X(T)+\sum_{i\in I} Y_i(T))\,|\,\mathcal{F}_t] \\
&=\Lambda(T)\exp\left(\mu(T-t)+X(t)+ \sum_{i\in I}\e^{-\beta_i(T-t)}Y_i(t)\right)\E\left[\exp(\sigma(B(T)-B(t))\right] \\
&\qquad\qquad\times\E\left[\exp\left(\sum_{i\in I}\sigma_i\int_t^T\e^{-\beta_i(T-s)}\,dB_i(s)\right)\right] \\
&=h(t,T)\exp\left(X(t)+\sum_{i\in I}\e^{-\beta_i(T-t)}Y_i(t) \right)\,.
\end{align*}
Then \eqref{eq:fdyn} follows straightforward from   It\^o's formula.
\end{proof}

 Note that for a deterministic interest rate, the forward and futures price process coincide (see for example Bj\"ork \cite{B}, Proposition 29.6), such that we will not distinguish between the two. \\
  We move on to option pricing. Consider a European call option with exercise time $\tau \leq T$ and strike $K$ written on the forward contract with dynamics given in Proposition \ref{prop:forward}.  The no-arbitrage price of the option is then defined as
  \begin{align*}
  C_I(t; \tau, T)=e^{-r(\tau-t)} \mathbb{E}[\max\{f(\tau,T)-K,0\}|\mathcal{F}_t]\;,
  \end{align*}
 where we price the option under the same measure $\mathbb{Q}$ as the forward. 
Using It\^o's formula   we see from \eqref{eq:fdyn} that we can write the forward price at time $\tau$ as
   \begin{align*}f(\tau,T)= f(t,T)\exp\{Z(t,\tau,T)\}\;,
  \end{align*}
  where 
   \begin{align}\label{eq:Z}  Z(t,\tau,T)&=\sigma (B(\tau)-B(t)) - \frac12 \sigma^2(\tau-t)\\
   &\qquad + \sum_{i\in I}\left(\int_t^\tau e^{-\beta_i(T-s)}\sigma_i dB_i(s) - \frac12 \int_t^\tau e^{-2\beta_i(T-s)}\sigma_i^2 ds\right)\nonumber\end{align}
is normally distributed with variance
\begin{align*}
\mbox{Var}(Z(t,\tau,T))&=  \sigma^2\mathbb{E}[ (B(\tau)-B(t))^2] + \sum_{i\in I}\mathbb{E}\left[\left(\int_t^\tau e^{-\beta_i(T-s)}\sigma_i dB_i(s)\right)^2\right]\\
&=\sigma^2(\tau-t)+  \sum_{i\in I}\int_t^\tau e^{-2\beta_i(T-s)}\sigma_i^2 ds\;.
\end{align*}
Then $ Z(t,\tau,T)$ has  standard deviation
  \begin{align}\label{eq:sigma}\sigma_I( T):=\sqrt{ \sigma^2(\tau-t)
  +\sum_{i\in I} c_ie^{-2\beta_i(T-\tau ) } }\;,
	\end{align} where 
	\begin{align}\label{eq:ci}c_i:= \frac{\sigma_i^2}{2\beta_i}\left( 1- e^{-2\beta_i(\tau -t) }  \right)\;.
	\end{align}
  
Hence,  we can immediately apply Blacks' Formula for options on forwards (see Black \cite{Black}, Benth et al. \cite{BSBK-book}, Chapter 9.1.1 for the mean reverting case)
   and find
\begin{proposition}\label{prop:black}
Let $f(t,T)$ be the forward price at time $t$  with delivery at $T$ given by Proposition \ref{prop:forward}. Then price of an European call option at time $t$ with  excercise time $\tau\leq T$ and strike price $K>0$ is given by
\begin{align}\label{eq:C}
	C_I(t; \tau, K, T)=e^{-r(\tau-t)} \left\{ f(t,T) \Phi(d_{1,I}) - K \Phi(d_{2,I})\right\}\;.
	\end{align}
	Here, 
	\begin{align*}
	d_{1,I}&= d_{2,I} + \sigma_I( T)\;,\\
	d_{2,I}&= \frac{\ln\left(\frac{f(t,T)}{K}\right) - \frac12 \sigma^2_I(T)}{\sigma_I(T)}\;,
	\end{align*}
	and $\Phi$ is the cumulative standard normal probability distribution function. 
	\end{proposition} 
	In the following chapter we analyse the role of the mean reverting components in Black's formula.

\section{The effect of mean reversion to the option price}\label{ch:pricingerror}
Electricity forwards deliver the underlying energy usually over a period $[T_1,T_2]$, as for example a month, a quarter of a year or a year.  The delivery period smooths the effect of the mean reversion in the forward price and should thus also lower its influence in the option price. This effect depends on the speed of mean reversion and the length of the delivery period.
 In our definition \eqref{eq:forward} of the forward price, we did not model the delivery period explicitly. Nevertheless, we choose  an alternative approach to take into account the delivery period, that is suitable especially in combination with options. Namely,  we  approximate the delivery period by its midpoint $T=\frac12(T_1 + T_2)$ and model the period by this point of time. As options on forwards are exercised at the beginning or shortly before the start of the delivery period, the length of the delivery period is given by  $2(T-\tau)$, the time from exercise of the option until the midpoint of the delivery period. For example, for a monthly delivery period with 30 days, we have that $T-\tau=15$ days, for forward that deliver the energy over a quarter of a year, it is $T-\tau= 45$ days. We will consider $T\to \infty$, that is we consider  delivery periods with increasing length.\\
Another approach is to model the delivery period explicitly and define the forward as expected average spot over this period, say \begin{align}\label{eq:avforward}F(t,T_1,T_2)=\mathbb{E}\left[\frac{1}{T_2-T_1}\int_{T_1}^{T_2}S(t) dt\Big|\mathcal{F}_t\right]\;.\end{align} The drawback of this definition is, that it does not lead to analytical closed form expressions for the forward price \eqref{eq:avforward}, when the spot price model is exponential.  For pricing options with Black's formula though, we need to consider an exponential framework. \\
We assume that the forward curve at time $t$ is given by \eqref{eq:fexpl}. 
Note that $f(t,T)$ goes to infinity for $T\to\infty$ due to the long term component $X$ in the underlying spot.  Therefore, we consider options where the strike $K$ is time dependent, too. That is, we let the relation between the strike and initial price of the underlying being proportional to each other:
\begin{align}\label{eq:d}
\frac{f(t,T)}{K(T)}=\delta
\end{align}
for some $\delta>0$. Like that, for $\delta=1$ we consider options at the money, for $\delta>1$ options in the money and for $\delta <1$ options out of the money.\\
We aim to examine the influence of one or more mean reverting components in the spot. Therefore, we consider additionally to \eqref{eq:spot} a spot price dynamics where we neglect some of the mean reverting components in $I$. For this purpose choose a subset $J\subset I$ of components and approximate \eqref{eq:spot} by 

\begin{align}\label{eq:spotJ}
S_J(t)=\Lambda(t) \exp\{X(t) +\sum_{i\in J} Y_i(t)\}\;.
\end{align}  The neglected mean reverting factors are then given by the set $I\setminus J$.
 Analogous to \eqref{eq:sigma} the corresponding  volatility of the log forward price $f_J(t,T)$ of \eqref{eq:spotJ}  is given by  
\begin{align}\sigma_J( T):= \sqrt{\sigma^2(\tau-t)
  +\sum_{i\in  J} c_ie^{-2\beta_i(T-\tau ) }} \;.
 \end{align} 
Let the price of an option $C_J(t; \tau, T)$ with underlying mean reverting factors $J$ be given by Black's formula as in Proposition \ref{prop:black}, where $\sigma_J( T)$ replaces  $\sigma_I( T)$:
\begin{align}\label{eq:CJ}
	C_J(t; \tau, T):=e^{-r(\tau-t)} \left\{ f_I(t,T) \Phi(d_{1,J}) - K \Phi(d_{2,J})\right\}\;,
	\end{align}
	where 
	\begin{align*}
	d_{1,J}&= d_{2,J} + \sigma_J( T)\;,\\
	d_{2,J}&= \frac{\ln\left(\frac{f(t,T)}{K}\right) - \frac12 \sigma^2_J(T)}{\sigma_J(T)}\;.
	\end{align*}
Here, we keep the initial curve to be the same, that is given by \eqref{eq:fexpl},  and add the index $I$ compared to \eqref{eq:fexpl} to emphasise the dependence on the set of mean reverting factors $I$. The initial curve is explicitly given at time $t$, while the volatility has to be estimated. Therefore it it is reasonable to consider options with the same initial curve, but different volatilities. Also, the more refined structure for the forward curve implied  by more mean reverting components is assumed to be the most fitting one. Hence, in the option price the chosen mean reverting components are reflected by the different volatilities $\sigma_I(T)$ and $\sigma_J(T)$ only.
If we do not consider any of the mean reverting components $Y_i$, $i\in I$, the corresponding forward price dynamics \eqref{eq:fdyn} reduces to a geometric Brownian motion
\begin{align}\label{eq:fdynB76}
  df_{B}(t,T)= f_{B}(t,T)\sigma dB(t) 
\end{align}
  and  we define $$\sigma_{B}:= \sqrt{\sigma^2(\tau-t)}\;.$$ Note  that  $\sigma_{I}( T)$ as well as $\sigma_{J}( T)$ converge to $\sigma_{B}$ as $T\to \infty$,
 giving an indication that for option pricing, the longer the delivery period, the more important is the non-stationary long term factor $X$ compared to the mean reverting factors $Y_i$, $i\in I$.  
Furthermore, it is possible to bound the volatilities by constants independent from   $T$
\begin{align}\label{eq:boundedsigma} 
	\sigma_B^2 \leq \sigma^2_I( T) \leq \sigma_B^2 + \sum_{i\in I} c_i\;,
\end{align}
with $c_i$ given as in \eqref{eq:ci}. 
	
We move on to state the upper and lower error estimates of the difference of the option price with set of mean reverting components $I$ and $J$.  For simplicity, we set $r=0$. 

\begin{proposition}\label{th:pricingerror}
Let $C_I$ be given by Proposition \ref{prop:black} and recall $C_J$ from \eqref{eq:CJ}.
For $\tau\leq T$ we have that
\begin{align*}
\alpha\sum_{i\in I\setminus J} c_i e^{ -( 2\beta_i-b)(T-\tau) }\leq C_I(t;\tau, T)-C_J(t;\tau,T)\leq \gamma\sum_{i\in I\setminus J} c_i e^{ -( 2\beta_i-b)(T-\tau) }\;,
\end{align*}
where
\begin{align*}
b&=\mu + \frac12 \sigma^2\;,\\
\alpha&=  \Lambda_l\frac{\delta}{2 \sqrt{2 \pi (\sigma_B^2 +\sum_{i \in I} c_i)}}\exp\left\{b (\tau-t) - \frac12 \left( \frac{|\ln \delta|}{\sigma_B} + \frac12 \sqrt{\sigma^2_B + \sum_{i\in I} c_i}\right)^2 +X(t) +\sum_{i\in I^-} Y_i(t)\right\}\;,\\
\gamma&=  \Lambda_u \frac{\delta}{2 \sqrt{2 \pi \sigma_B^2 }} \exp\left\{b (\tau-t) +X(t) +\sum_{i\in I^+} Y_i(t)\right\} \;, 
\end{align*}
and $\label{eq:I}I^-=\{i\in I: Y_i(t)\leq0\}$ as well as $I^+=\{i\in I: Y_i(t)>0\}$.
\end{proposition}	
\begin{proof}
Consider Blacks' option price from Proposition \ref{prop:black} as a function of the variance of the underlying:
\begin{align}\label{eq:d1}
C(z)&=f(t,T) \Phi(d_1(z)) - K\Phi(d_2(z))\;,\qquad\text{where}\nonumber\\
d_1(z)&= d_2(z)+\sqrt z,\\
d_2(z)&= \frac{\ln( \frac{f(t,T)}{K}) -\frac12 z}{\sqrt z}\nonumber\;.
\end{align}
Using that $d_1(z)= d_2(z)+\sqrt z$, we find (analogously to the determination of the vega of an option) 
\begin{align*}
\frac{dC}{d z}= \frac{1}{2 \sqrt z } f(t,T) \Phi'(d_1(z))>0\;, 
\end{align*} 
where $\Phi'(x)= \frac{1}{\sqrt{2\pi}}\exp\{-\frac12 x^2\}$. 
Thus, $C$ is continuous in $z>0$ and strictly increasing and with the mean value theorem it follows the existence of a $\xi \in (\sigma^2_J( T), \sigma^2_I( T))$ such that 
\begin{align}\label{eq:mvt}C(\sigma^2_I( T))-C(\sigma^2_J( T))
= \frac{dC}{d z}\left(\xi\right) \left(\sigma^2_I(T)-\sigma^2_J( T)\right)
= \frac{1}{2 \sqrt \xi }  \Phi'(d_1(\xi))f(t,T)\left(\sigma^2_I(T)-\sigma^2_J( T)\right)\;. \end{align} First, note that $$\sigma^2_I(T)-\sigma^2_J( T)=  \sum_{i\in I\setminus J} c_i e^{ - 2\beta_i(T-\tau) }\;.$$
Using that $$\sigma_B^2 \leq \sigma^2_J( T)<\xi < \sigma^2_I( T)\leq \sigma_B^2 +\sum_{i\in I} c_i$$ with $c_i$ as in \eqref{eq:ci} we find that
$$ \frac{1}{2\sqrt{\sigma_B^2 +\sum_{i\in I} c_i}} <\frac{1}{2\sqrt{\xi}}< \frac{1}{2\sqrt{\sigma_B^2 }}$$ and
$$ \frac{1}{\sqrt{2\pi}}\exp\left\{ -\frac12 \left( |\ln \delta| \frac{1}{\sqrt{\sigma_B^2}}+ \frac12 \sqrt{\sigma_B^2 +\sum_{i\in I} c_i}\right)^2\right\}\leq \Phi'(d_2(\xi)) \leq \frac{1}{\sqrt{2\pi}}\;.$$ 
Let 
\begin{align}\label{eq:I}I^-=\{i\in I: Y_i(t)\leq0\}\text{ and }  I^+=\{i\in I: Y_i(t)>0\}\;.
\end{align}
Then we find for the initial curve \eqref{eq:fexpl}
$$ f(t,T) \geq  \Lambda_l \exp\left\{ b(\tau-t)  +X(t) +\sum_{i\in I^-} Y_i(t)\right\}\exp\left\{- b(T-\tau)\right\}$$ and
$$ f(t,T) \leq  \Lambda_u \exp\left\{ b(\tau-t) +\frac12 \sum_{i\in I}\frac{\sigma_i^2}{2 \beta_i} +X(t) +\sum_{i\in I^+} Y_i(t) \right\}\exp\left\{ -b(T-\tau)\right\}\;,$$ where $\Lambda_l$ and $\Lambda_u$ denote the lower and upper bounds of $\Lambda$.
Collecting the terms the result follows. 
\end{proof}

Proposition \ref{th:pricingerror} gives  an analytical formula for the error's  asymptotic behaviour. It is a weighted sum of  exponentials,  determined through the length from exercise of the option  until delivery of the forward, the speed of the neglected mean reverting components and the characteristics of the long term component. The difference between the option prices is fundamentally determined through the differences of the variances. In fact, we have that
\begin{align*}
C_I(t;\tau, T)-C_J(t;\tau,T)\sim c \exp\{ b(T-\tau)\}\left(\sigma^2_I(T)-\sigma^2_J( T)\right)\;
\end{align*}
for $T$ big. \\
Although the volatility is the only parameter that varies in the options prices $C_I(t;\tau, T)$ and  $C_J(t;\tau, T)$ and 
$$ \left|\sigma_I(T)-\sigma_J( T)\right|\to 0 \qquad \text{for}\qquad T\to \infty\;,$$ the difference between the corresponding option prices  does not always converge to zero for $ T\to \infty$ if there is a non-stationary component $X$ in the spot. In this case, the initial forward curve converges to infinity with exponential rate. Now, looking at \eqref{eq:mvt} the difference of the variances converges to zero  with exponential rate, too, such that the convergence depends on what rate is larger.  
 If for all neglected mean reverting components $i \in I\setminus J$ we have that $2\beta_i-(\mu + \frac12 \sigma^2)>0$ then the difference between the two option prices converges to zero.

We assume the spot price dynamics to be given immediately under the pricing measure $\mathbb{Q}$ for the forward. In our setting, typically a Girsanov transform is used to determine this pricing measure, which changes the drift of the driving Brownian motions under $\mathbb{Q}$ compared to the physical measure $\mathbb{P}$. For the base component $X$ as in \eqref{eq:X}  this results in a change of the drift, captured in the $\mu$. Thus, in Proposition \ref{th:pricingerror} it influences  the  speed of the convergence. For the stationary factors $Y_i$, $i\in I$ as in \eqref{eq:Y} a change of measure influences the mean reversion level, which has impact on the constants $\alpha$ and $\gamma$  in Proposition \ref{th:pricingerror}.

   \begin{figure}[ht]
        \centering
        \begin{minipage}{.5\textwidth}
        \centering
     \includegraphics[width=0.9\linewidth]{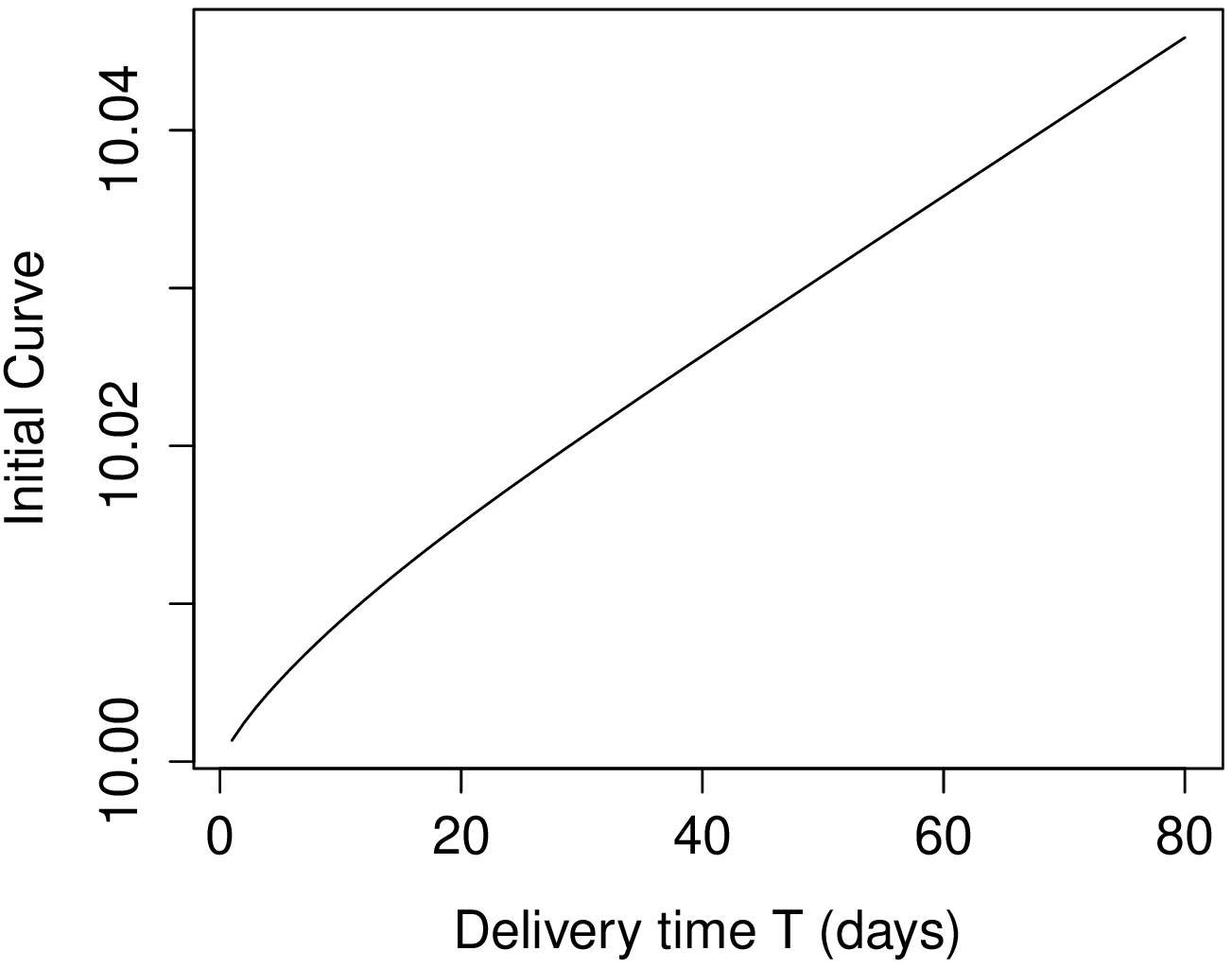}
       \caption{Initial curve. }
      \label{fig:InitialCurve}
        \end{minipage}%
        \begin{minipage}{.5\textwidth}
            \centering
            \includegraphics[width=0.9\linewidth]{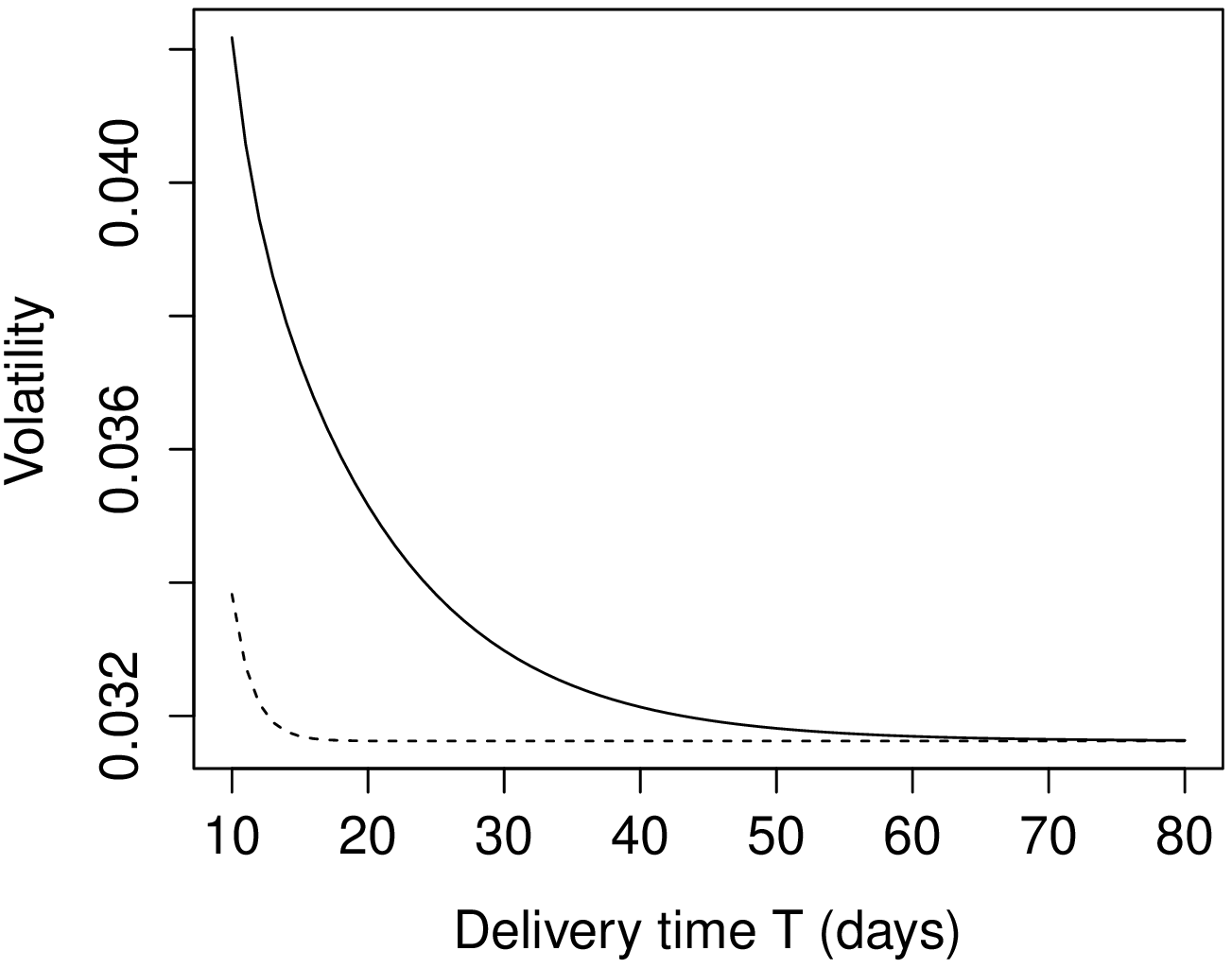}
               \caption{$\sigma_I(T)$ and $\sigma_J(T)$ (dashed).}
                    \label{fig:Volatility}
         \end{minipage} 
  \end{figure} 
  
We illustrate our result with an example and consider  options at the money, that is $\delta=1$, at time $t=0$ with exercise at $\tau= 10$.  Assume that there are two mean reverting components $I=\{1, 2\}$, where we have a fast mean reverting component with $\beta_1=0.3466$, corresponding to a half life of two days, and a slowly mean reverting component with $\beta_2=0.0495$ corresponding to a half life of two weeks. We want to illustrate the pricing error implied by the slowly mean reverting component, such that we choose to model the fast mean reverting component only and $J=\{1\}$. Furthermore, let $\sigma_1=\sigma_2=\sigma=0.01$,  $\mu=0$ and   $X(0)=Y_1(0)=Y_2(0)=0$ as well as $\Lambda(T)=10$ for all $T$. In Figures \ref{fig:InitialCurve} and 2 the initial curve $f(0,T)$  and the volatilities $\sigma_I(T)$ and $\sigma_J(T)$ are depicted for $T=0,\dots, 80$. Figure \ref{fig:PricingError}   then shows the absolute pricing error together with its lower and upper bound and we find that our bounds are very sharp.  The relative pricing error is shown in Figure \ref{fig:RelativePricingError}. Here we see that for a delivery period of a month, corresponding to $T=25$, we have a relative pricing error of 7.1 \%, which is quite a heavy mispricing. For a delivery period of a quarter of a year, corresponding to $T=55$, we find that we only misprice by $0.2\%$.

  \begin{figure}[ht]
       \centering
       \begin{minipage}{.5\textwidth}
       \centering
    \includegraphics[width=0.9\linewidth]{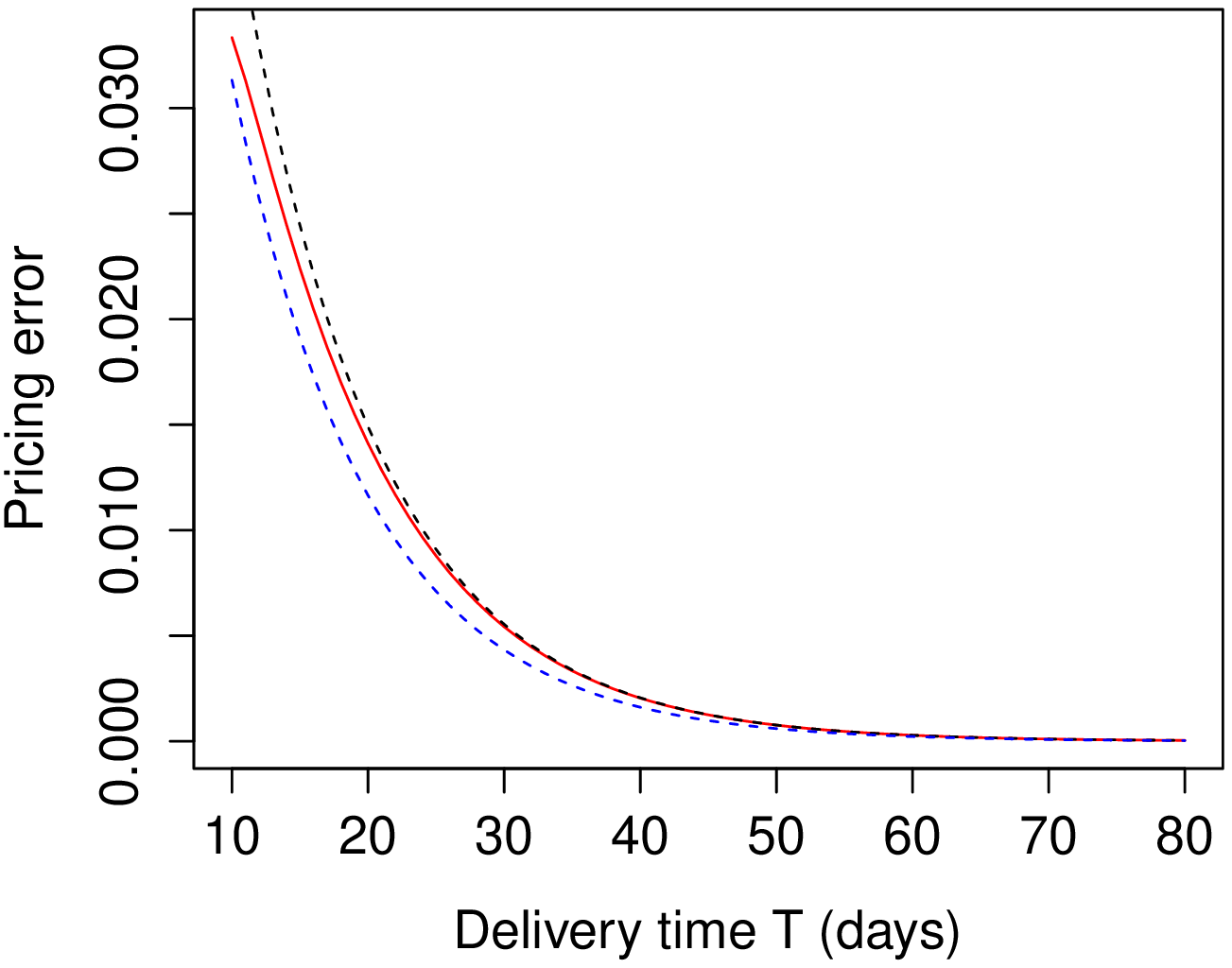}
      \caption{Pricing error and bounds (dashed).    }
   \label{fig:PricingError}
       \end{minipage}
       \begin{minipage}{.5\textwidth}
           \centering
           \includegraphics[width=0.9\linewidth]{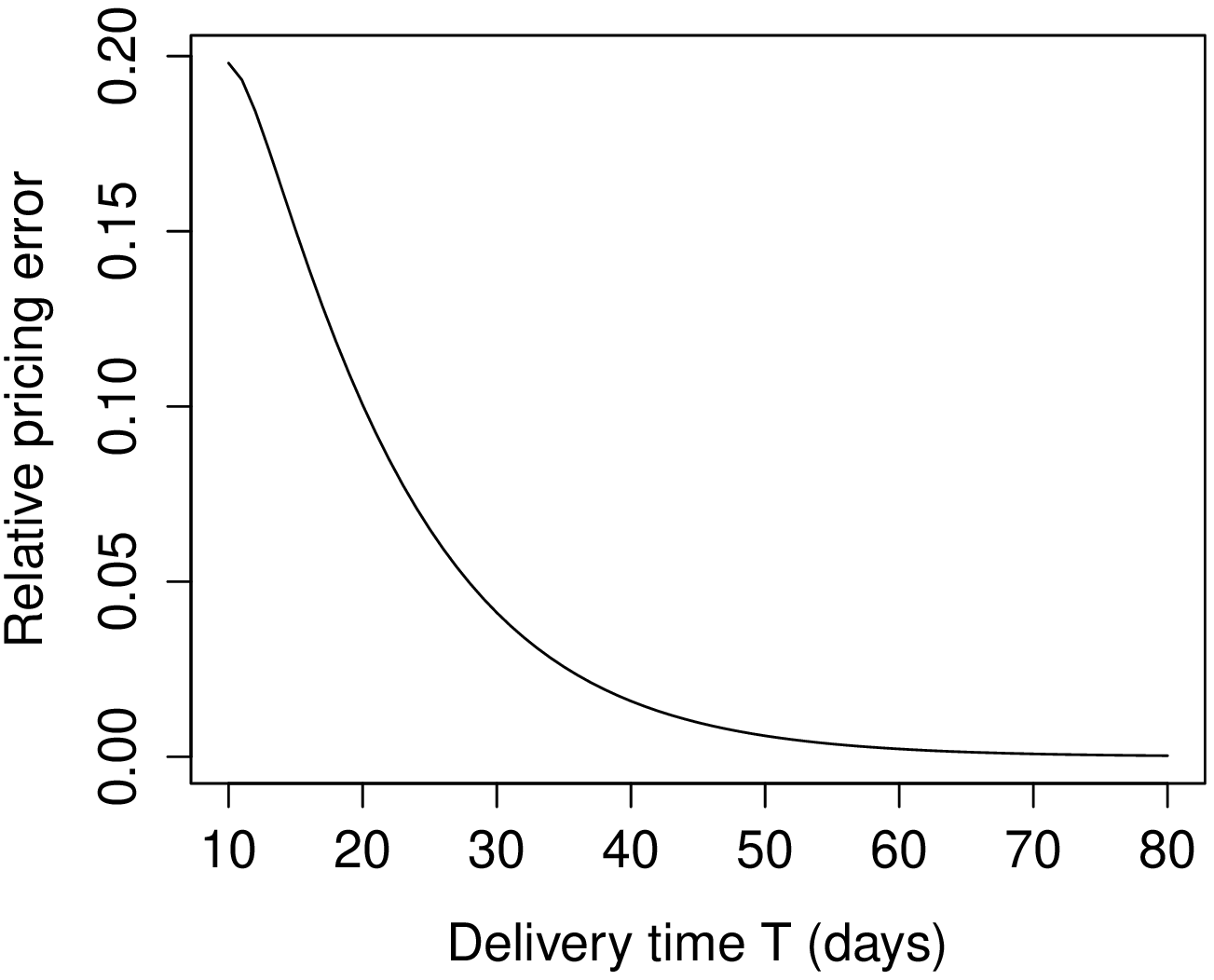}
              \caption{Relative pricing error. }
                 \label{fig:RelativePricingError}
        \end{minipage} 
 \end{figure} 

\section{The effect of mean reversion to hedging}\label{ch:hedging}
For option prices given by Black's formula as in Proposition \ref{prop:black}, the delta hedge, the derivative of the option price with respect to the initial value, is given by 
$$\Delta_I(t;\tau,T)=e^{-r(\tau-t)}\Phi(d_{1,I}) \;,$$ where 
	\begin{align*}
	d_{1,I}&= d_{2,I} + \sigma_I( T)\;,\\
	d_{2,I}&= \frac{\ln\left(\frac{f_I(t,T)}{K}\right) - \frac12 \sigma^2_I(T)}{\sigma_I(T)}\;
	\end{align*}
(see for example Benth et al. \cite{BSBK-book}).
As for the option price, we aim to quantify the influence of the mean reverting components to the hedging strategy depending on the length of the delivery period of the underlying future. 
Therefore, we derive upper and lower  bounds for the error we commit if we take into account a set of mean reverting components $I$ or only a subset $J\subset I$. Let the delta hedge  $\Delta_J(t;\tau, T)$  then be given by
$$\Delta_J(t;\tau,T)=e^{-r(\tau-t)}\Phi(d_{1,J}) \;,$$ where 
\begin{align*}
d_{1,J}&= d_{2,J} + \sigma_J( T)\;,\\
	d_{2,J}&= \frac{\ln\left(\frac{f_I(t,T)}{K}\right) - \frac12 \sigma^2_J(T)}{\sigma_J(T)}\;.
 	\end{align*}
 	As in the previous chapter, the initial curve is given by $f_I(t,T)$, as the initial curve is observable at time $t$ and we assume that the more refined shape given by the mean reverting components in $I$ leads to  a better description of the initial curve. Consequently, the hedging strategies $\Delta_I(t;\tau, T)$ and $\Delta_J(t;\tau, T)$ differ in the volatility only.
  We find the following
\begin{proposition}\label{th:hedgingerror}
For $2\ln(\delta)\leq \sigma_B^2$ or $2\ln(\delta)\geq \sigma_B^2+ \sum_{i\in I} c_i$ the hedging error can be bounded by
\begin{align*}
h\sum_{i\in I\setminus J} c_i e^{ -2\beta_i(T-\tau)}\leq | \Delta_I(t;\tau,T) - \Delta_J(t;\tau,T)| \leq g\sum_{i\in I\setminus J} c_i e^{ -2\beta_i(T-\tau)}\;,
\end{align*}
where
\begin{align*}
g= \begin{cases}
\frac{1}{4\sqrt{2\pi}}\left(\sigma_B^2\right)^{-\frac32}\left|\sigma_{B}^2+\sum_{i\in I}c_i - 2\ln(\delta)\right|\;,& \text{ if }\; 2 \ln(\delta) \leq \sigma_B^2\;,\\
\frac{1}{4\sqrt{2\pi}}\left(\sigma_B^2\right)^{-\frac32}\left|\sigma_{B}^2 - 2\ln(\delta)\right|\;,& \text{ if }\; \sigma_{B}^2+\sum_{i\in I}c_i  \leq 2 \ln(\delta),
\end{cases}
\end{align*}
and
\begin{align*}
h=
\begin{cases}
 \frac{k}{4\sqrt{2\pi}}\left(\sigma_{B}^2+\sum_{i\in I}c_i\right)^{-\frac32}\left|\sigma_{B}^2 - 2\ln(\delta)\right|\;,& \text{if } 2 \ln(\delta) \leq \sigma_B^2\;,\\
\frac{k}{4\sqrt{2\pi}}\left(\sigma_{B}^2+\sum_{i\in I}c_i\right)^{-\frac32}\left|\sigma_{B}^2+\sum_{i\in I}c_i - 2\ln(\delta)\right|\;,&\text{if }  \sigma_{B}^2+\sum_{i\in I}c_i  \leq 2 \ln(\delta)\,,
\end{cases}
\end{align*}
with
$$k=\exp\left\{-\frac12\left(\ln(\delta)^2\frac{1}{\sigma_B^2}+|\ln(\delta)|+ \frac14\left(\sigma_{B}^2+\sum_{i\in I}c_i\right)\right)\right\}\,,$$ and $c_i$, $i\in I$ given by \eqref{eq:ci}.\end{proposition}
\begin{proof}
Using the mean value theorem it follows that
\begin{align*}
|\Delta_I(t;\tau,T)-\Delta_J(t;\tau,T)| = \left|\frac{d\Phi(d_1(z))}{dz}\Big|_{z=\xi}\right| |\sigma_I^2(T)-\sigma_j^2(T)|
\end{align*} for some $\xi \in (\sigma_J^2(T),\sigma_I^2(T))$ and where $d_1(z)$ is defined by \eqref{eq:d1}. We have $\sigma_B^2 \leq \sigma_J^2(T)\leq \sigma_I^2(T) \leq \sigma_B^2 +\sum_{i\in I}c_i$ and therefore
\begin{align}\label{eq:volaxi}
\sigma_B^2 < \xi < \sigma_B^2 +\sum_{i\in I}c_i\;.
\end{align}
It is $$\frac{d\Phi(d_1(z))}{dz}=\Phi'(d_1(z)) d_1'(z)=\frac{1}{4\sqrt{2\pi}}\exp\left\{-\frac12 d_1^2(z)\right\}z^{-\frac32}[z-2\ln(\delta)]\;.$$
Now with the binomic formula, the triangle inequality and \eqref{eq:volaxi} it follows that 
\begin{align*}d_1^2(\xi)=\left|\ln(\delta) \frac{1}{\sqrt{\xi}} +\frac12 \sqrt{\xi}\right|^2\leq \ln(\delta)^2 \frac{1}{\sigma_B^2} +|\ln(\delta)|+ \frac14 (\sigma_B^2+\sum_{i\in I}c_i)\;.
\end{align*}
For $\ln(\delta)<0$ as well as for $0\leq 2\ln(\delta)\leq\sigma^2_B$ the term $(z-2\ln(\delta))$ is positive for all $z\in[\sigma_B^2, \sigma_B^2 + \sum_{i\in I}c_i]$. For $2\ln(\delta)>\sigma_B^2+\sum_{i\in I}c_i$ the term $(z-2\ln(\delta))$ is negative and increasing  in $z\in[\sigma_B^2, \sigma_B^2 + \sum_{i\in I}c_i]$ and so $|z-2\ln(\delta)|$ is positive and decreasing in $z$. The statement follows using \eqref{eq:volaxi} and collecting the terms. 
\end{proof}
We illustrate these results by continuing the example at the end of the previous chapter. The corresponding hedging error, the upper and lower bound for the error are shown in Figure \ref{fig:hedg}. Figure \ref{fig:hedg3} then shows the relative hedging error.
  \begin{figure}[ht]
       \centering
       \begin{minipage}{.5\textwidth}
       \centering
      \includegraphics[width=.9\linewidth]{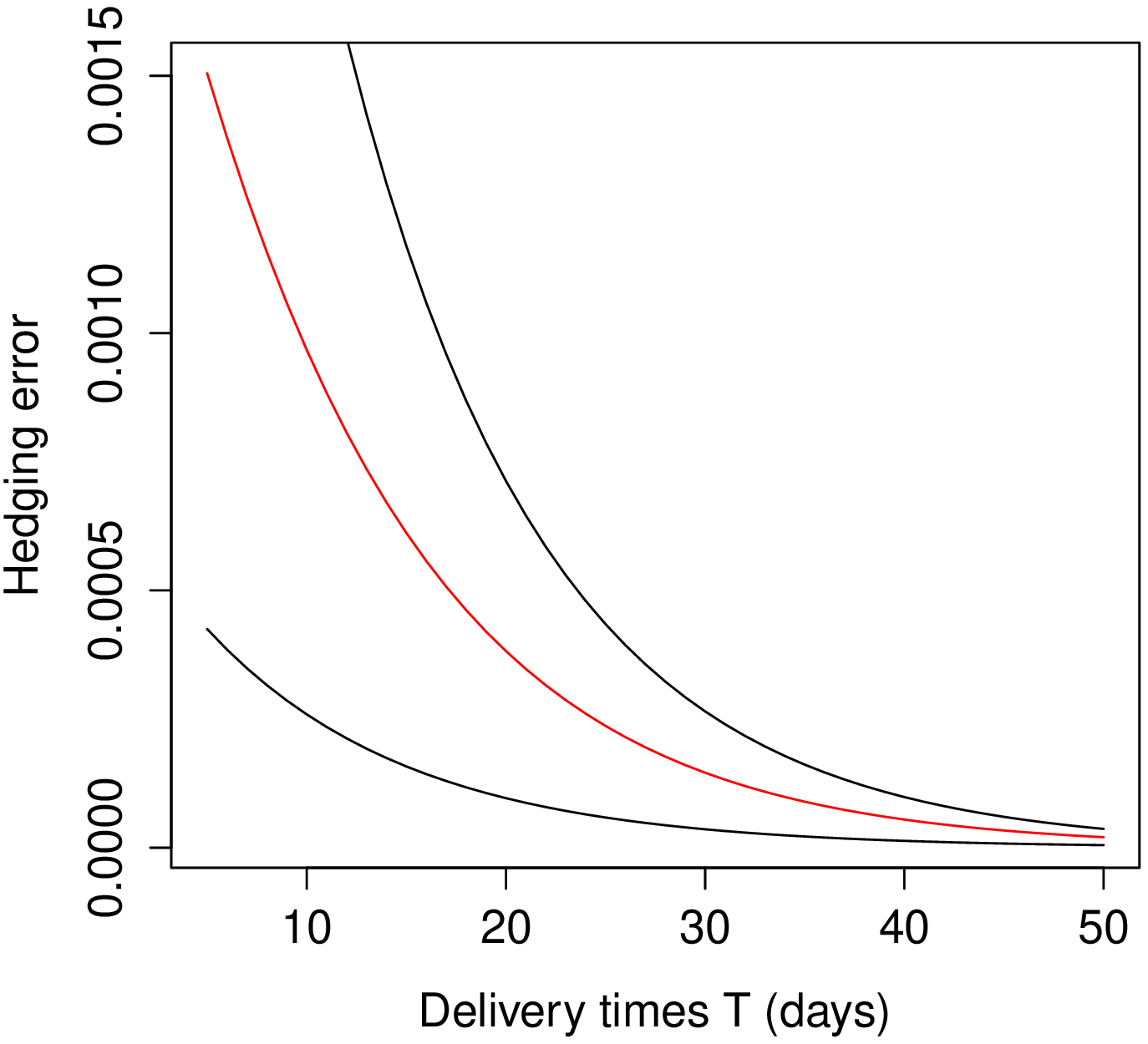}
      \caption{Hedging error and bounds for $\ln(\delta)=0$. }
       \label{fig:hedg}
       \end{minipage}%
       \begin{minipage}{.5\textwidth}
           \centering
           \includegraphics[width=.9\linewidth]{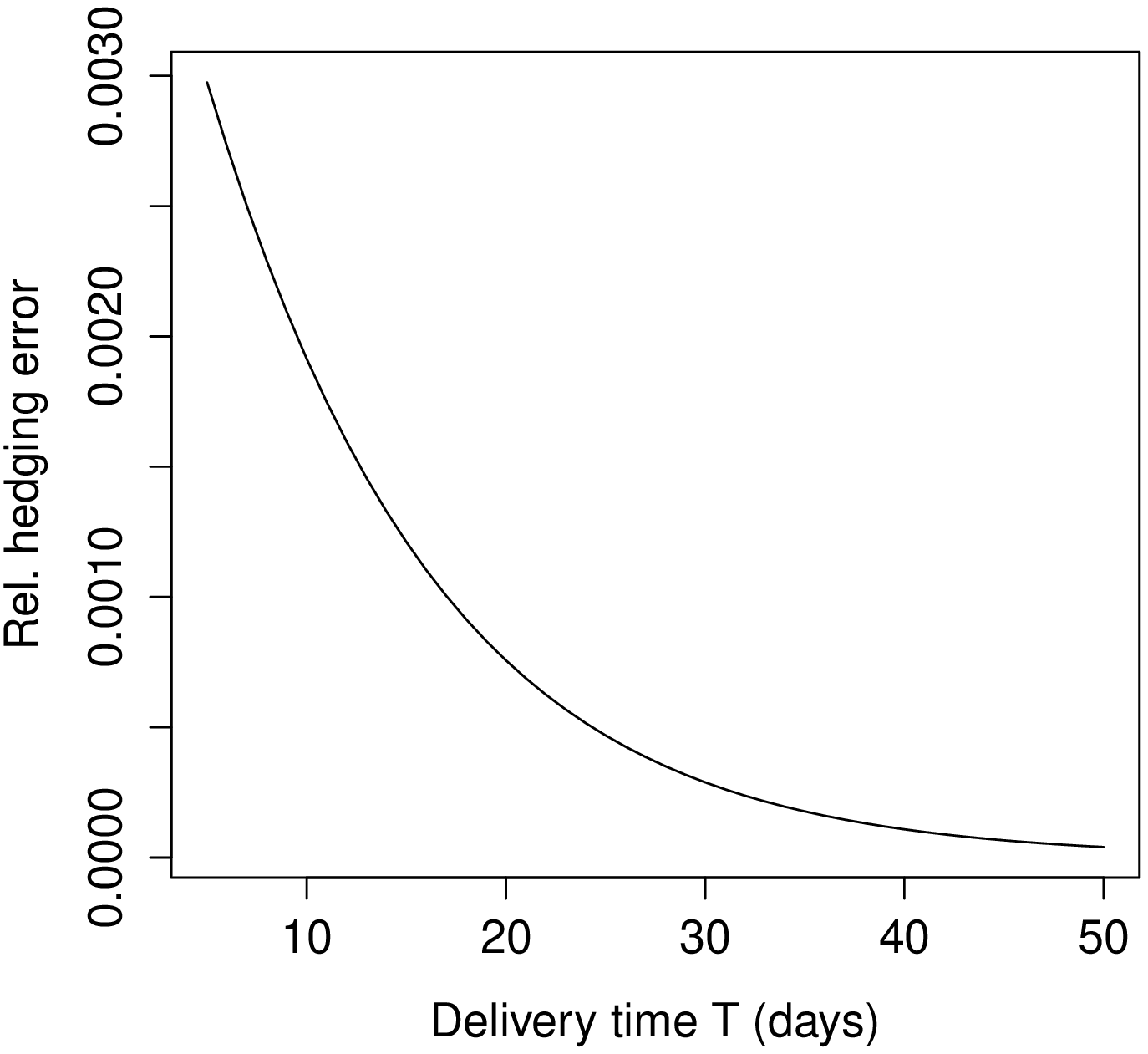}
           \caption{Relative hedging error  for $\ln(\delta)=0$. }
            \label{fig:hedg3}
        \end{minipage} 
 \end{figure} 
   
\section{Conclusion}
A special feature of electricity spot prices is its  multi-scale mean reverting behaviour. While the fast mean reverting spikes receive a remarkable amount of attention in the literature, the existing  slow mean reversion in the electricity spot has been treated less so far. In this paper, we focus on the role of slow mean reversion when pricing and hedging options on forwards. In particular, we find upper and lower bounds for the error one makes neglecting mean reverting factors. The error bounds are of exponential shape, where the speed of the exponential decay is determined through the speed of the neglected mean reverting factors and  the parameters of the long term factor. Consequently, for a slow speed of mean reversion it takes some time until the minimum error becomes small, while for a fast mean reversion rate, the error is insignificant very fast. In fact, we find that the error in the option price behaves asymptotically as   the  volatility corresponding to the neglected components times the exponential growth of the initial curve. We find similar results for the hedging component, where the error bounds decrease in terms of the neglected volatility only.  \\
The combination of  an exponential framework for the spot and an explicit modelling of the delivery period of the forward does not lead to closed form solutions for the forward prices. In this paper, we decide to keep the exponential framework as it leads to Black's formula for pricing options on forwards - which is indeed used in industry - and we choose to model the delivery period by its midpoint. Future work includes a similar analysis as carried out in this paper, but where the delivery period is modelled explicitly in an arithmetic framework.

\end{document}